\documentclass[12pt, a4paper]{article}

\usepackage[T1]{fontenc}
\usepackage[latin1]{inputenc}

\usepackage{amsmath}
\usepackage{amssymb}
\usepackage{amsthm}

\usepackage[english]{babel}

\usepackage{hyperref}
\usepackage[numbers]{natbib}
\newtheorem{theorem}{Theorem}
\newtheorem{lemma}[theorem]{Lemma}
\theoremstyle{definition}
\newtheorem*{remark}{Remark}
\newtheorem*{definition}{Definition}

\renewcommand{\d}{\mathrm{d}}
\newcommand{\cL}{\mathcal{L}}

\newcommand{\var}[3]{\frac{\delta_{#1} #2}{\delta #3}}
\newcommand{\der}[2]{\frac{\partial #1}{\partial #2}}

\begin{document}

\title{Lagrangian multiforms and dispersionless integrable systems}
\author{Evgeny V. Ferapontov\footnote{ORCID 0000-0001-9025-9478, \texttt{E.V.Ferapontov@lboro.ac.uk}}\ \ and Mats Vermeeren\footnote{ORCID 0000-0002-6982-4505, \texttt{M.Vermeeren@lboro.ac.uk} (corresponding author)} }
   \date{}
   \maketitle
   \vspace{-7mm}
\begin{center}
Department of Mathematical Sciences \\ Loughborough University \\
Loughborough, Leicestershire LE11 3TU \\ United Kingdom
\end{center}

\bigskip

\begin{abstract}
\noindent
We demonstrate that interesting examples of Lagrangian multiforms appear naturally in the theory of multidimensional dispersionless  integrable systems as (a)  higher-order conservation laws of linearly degenerate PDEs in 3D, and (b) in the context of Gibbons-Tsarev equations governing hydrodynamic reductions of heavenly type equations in 4D.


\end{abstract}

\vfill

{\footnotesize\noindent This version of the article has been accepted for publication, after peer review but is
	not the Version of Record and does not reflect post-acceptance improvements, or any corrections. The Version of Record is
	available online at: \url{https://doi.org/10.1007/s11005-025-02016-w}}

\pagebreak
\tableofcontents

\section{Lagrangian multiform formalism}
\label{sec:1}

Lagrangian multiform theory provides a variational principle for systems that are integrable in the sense of multidimensional consistency. It applies in the discrete case (where multidimensional consistency takes the form of commuting maps, consistency around the cube, etc), in the continuous case (where multidimensional consistency takes the form of commuting flows or involutive PDEs), and in semi-discrete settings.

The central object in continuous Lagrangian multiform theory is a differential $d$-form. Typically, a Lagrangian $d$-form describes a system of equations in $d$ independent variables each. For example,
\begin{itemize}
	\item a continuous Lagrangian 1-form describes commuting ODEs and can often be associated to a system of Poisson-commuting Hamiltonians \cite{yoo-kong2011discretetime, suris2013variational};
	\item a continuous Lagrangian 2-form describes a hierarchy of 2D PDEs, such as the (potential) Korteweg-de Vries hierarchy \cite{suris2016variational,suris2016lagrangian};
	\item continuous Lagrangian 3-forms describe hierarchies of 3D PDEs, such as the KP hierarchy \cite{sleigh2023lagrangian, nijhoff2023lagrangian}.
\end{itemize}
This does not mean that all the Euler-Lagrange equations of a Lagrangian $d$-form are $d$-dimensional PDEs, but rather that the system of Euler-Lagrange equations is of a $d$-dimensional nature. In this paper we will show examples of 2-forms for which the variational equations appear as a system of PDEs in 3 variables each. In addition, this paper establishes a new connection between Lagrangian multiforms and conservation laws, and provides examples of Lagrangian multiforms for Gibbons-Tsarev type systems.

\subsection{Multiform Euler-Lagrange equations}
We consider Lagrangian 2-forms
\begin{equation*} 
\cL = \sum_{i < j}L_{ij}\, \d x^i\wedge \d x^j.
\end{equation*}
Here $L_{ij}$ are some functions of the jet bundle variables $x^i, x^j, u, u_i, u_j, u_{ij}$ where $u$ is a function of the independent variables $x^1, \dots, x^n$ and $u_i=u_{x^i}, \ u_{ij}=u_{x^ix^j}$, etc, denote partial derivatives. We assume $L_{ji} = -L_{ij}$.
For a two-dimensional surface $\Gamma$ in the space of independent variables, we consider the action integral $S_\Gamma = \int_\Gamma \cL$. Lagrangian multiform theory requires that this action is critical for every choice of surface $\Gamma$, with respect to variations of $u$. (This requirement is also known as the \emph{pluri-Lagrangian} principle \cite{boll2014what, suris2016variational}.) This is the case if and only if the following three groups of \emph{multiform Euler-Lagrange} equations hold:
\begin{subequations}
\label{mfEL}
\begin{align}
	& \var{ij}{L_{ij}}{u} = 0 , \label{mfELa} \\
	& \var{ij}{L_{ij}}{u_j} - \var{ik}{L_{ik}}{u_k} = 0 , \label{mfELb} \\
	& \var{ij}{L_{ij}}{u_{ij}} + \var{jk}{L_{jk}}{u_{jk}} + \var{ki}{L_{ki}}{u_{ki}} = 0. \label{mfELc}
\end{align}
Here $\var{}{\cdots}{\cdots}$ denote variational derivatives which, under the assumption that $L_{ij}$ does not depend on third or higher derivatives of $u$, can be computed as
\begin{align*}
	& \var{ij}{}{u} = \der{}{u} - \partial_i \der{}{u_i} - \partial_j \der{}{u_j} +  \partial_i\partial_j \der{}{u_{ij}}, \\
	& \var{ij}{}{u_j} = \der{}{u_j} -\partial_i \der{}{u_{ij}}  , \\
	& \var{ij}{}{u_{ij}} = \der{}{u_{ij}},
\end{align*}
where $\partial_i = \der{}{x^i}$, etc. There is no summation over repeated indices.

If we allow the coefficients $L_{ij}$ to depend on derivatives with respect to $x^k$ for $k \neq i,j$ as well, but still assume no third or higher derivatives occur in $L_{ij}$, the system of multiform Euler-Lagrange equations \eqref{mfELa}--\eqref{mfELc} is extended by
\begin{align}
	\var{ij}{L_{ij}}{u_k} &= 0 \qquad \forall k \neq i,j , \label{mfEL-uk}\\
	\der{L_{ij}}{u_{k\ell}} &= 0 \qquad \forall k,\ell \neq i,j , \label{mfEL-ukl} \\
	\der{L_{ij}}{u_{j\ell}} - \der{L_{ik}}{u_{k\ell}} &= 0  \qquad \forall \ell \neq i . \label{mfEL-ujl}
\end{align}
\end{subequations}
Note that the definition of variational derivatives remains the same: the additional derivatives are only taken in the variables corresponding to the indices of $\delta$. In particular,
\[ \var{ij}{}{u_k} = \der{}{u_k} - \partial_i \der{}{u_{ik}} - \partial_j \der{}{u_{jk}} . \]

For this framework to be nontrivial, we need to impose the condition that the multiform Euler-Lagrange equations (also known as \emph{multi-time Euler-Lagrange equations}) are in involution, i.e.\@ that the overdetermined system \eqref{mfEL} admits solutions for
suitable generic boundary conditions.

The fundamental idea of Lagrangian multiforms was established in \cite{lobb2009lagrangian}. The multiform Euler-Lagrange equations for a 2-form $\cL$ depending on the second jet bundle first appeared in \cite{suris2016variational} and the multiform Euler-Lagrange equations for $\cL$ depending on any jet bundle were derived in \cite{suris2016lagrangian}. This derivation makes use of surfaces that are made up of flat pieces in coordinate directions. Then equation \eqref{mfELa} comes from the flat pieces, \eqref{mfELb} from the edges where two flat pieces meet, and \eqref{mfELc} from the corners where three or more flat pieces meet. For this reason, \eqref{mfELa} are called planar equations, \eqref{mfELb} edge equations and \eqref{mfELc} corner equations. Other derivations of the multiform Euler-Lagrange equations are possible. One can obtain equations (\ref{mfEL}) by restricting $\cL$ to an arbitrary plane in the space of independent variables (not necessarily  coordinate planes) and calculating the standard Euler-Lagrange equation, which should be satisfied for any such plane. Furthermore, in \cite{sleigh2023lagrangian} it was pointed out that equivalent equations are obtained by taking pointwise variations of $\d \cL$.

The connection between Lagrangian multiforms and conservation laws has been touched upon in \cite{suris2016variational, petrera2017variational} and was used in \cite{petrera2021variational, sleigh2020variational} to construct Lagrangian 2-forms for systems with known variational symmetries. These papers start with a known Lagrangian and use Noether's theorem to obtain conservation laws from which a multiform is then constructed. In the present work, we take conservation laws as our starting point and investigate the additional structure provided by Lagrangian multiform theory.

\subsection{Double zero property}

In many examples, the exterior derivative of $\cL$ factorises as
\begin{equation}
	\label{factorisation}
	\d \cL =\sum_{i<j<k} A_{ijk}B_{ijk}\, \d x^i \wedge \d x^j\wedge \d x^k,
\end{equation}
and it is observed that the system
\begin{equation}
	\label{ABsystem}
	A_{ijk}=0,\quad B_{ijk}=0
\end{equation}
is equivalent to the full system of multiform Euler-Lagrange equations \eqref{mfEL}. 
For example, for the Lagrangian multiform 
\begin{equation}\label{LMF}
	\cL = \sum_{i<j}(c^i-c^j)\frac{u_{ij}^2}{u_iu_j}\ \d x^i\wedge \d x^j,
\end{equation}
where $c^i=\text{const}$, the factorisation \eqref{factorisation} holds with
\begin{align*}
	A_{ijk} &= (c^i-c^j)\frac{u_{ij}}{u_iu_j}+(c^j-c^k)\frac{u_{jk}}{u_ju_k}+(c^k-c^i)\frac{u_{ik}}{u_iu_k}, \\
	B_{ijk} &= 2 u_{ijk} - \left(\frac{u_{ij}u_{ik}}{u_i}+\frac{u_{ij}u_{jk}}{u_j}+\frac{u_{ik}u_{jk}}{u_k} \right),
\end{align*}
while the multiform Euler-Lagrange equations  \eqref{mfEL} are
\begin{subequations} 
	\label{vweEL}
	\begin{align}
		&\!\left(\frac{u_{ij}^2}{u_i^2u_j}\right)_i+\left(\frac{u_{ij}^2}{u_iu_j^2}\right)_j+\left(\frac{2u_{ij}}{u_iu_j}\right)_{ij}=0, \label{vweEL1}\\
		&(c^j - c^i) \left( \frac{u_{ij}^2}{u_i u_j^2} + 2 \frac{u_{ii}u_{ij}}{u_i^2 u_j} - 2 \frac{u_{iij}}{u_i u_j}\right)  - (c^k - c^i) \left( \frac{u_{ik}^2}{u_i u_k^2} + 2 \frac{u_{ii}u_{ik}}{u_i^2 u_k} - 2 \frac{u_{iik}}{u_i u_k}\right) = 0,  \label{vweEL2}
		\\
		& (c^i-c^j)\frac{u_{ij}}{u_iu_j}+(c^j-c^k)\frac{u_{jk}}{u_ju_k}+(c^k-c^i)\frac{u_{ik}}{u_iu_k} = 0 . \label{vweEL3}
	\end{align}
\end{subequations}
This system is equivalent to the  equations $A_{ijk} = 0$, $B_{ijk} = 0$ (a fact which is not immediately obvious, see Appendix \ref{app-vwe}).
Equations $A_{ijk}=0$ form the so-called Veronese web hierarchy \cite{zakharevich2000nonlinear},
\begin{equation} \label{V}
	(c^i-c^j)u_ku_{ij}+(c^k-c^i)u_ju_{ik}+(c^j-c^k)u_iu_{jk}=0,
\end{equation}
while equations $B_{ijk}=0$,
\begin{equation} \label{E}
	u_{ijk}=\frac{1}{2}\left(\frac{u_{ij}u_{ik}}{u_i} +\frac{u_{ij}u_{jk}}{u_j}+\frac{u_{ik}u_{jk}}{u_k}\right),
\end{equation}
characterise potential (Egorov) metrics $\sum_i u_i (\d x^i)^2$ with \emph{diagonal curvature}, meaning that all curvature components $R^i_{kkj}$ with  distinct $i, j, k$ are identically zero. Equations (\ref{E}) have appeared in  \cite{schief2025affine} in the context of multidimensional consistency of partial differential equations.
Remarkably, equations (\ref{V}) and (\ref{E}), although coming from entirely different geometric contexts, are in involution. This can be checked by computing the Mayer bracket of $A_{ijk}$ and $B_{ijk}$ \cite{kruglikov2002mayer}. Furthermore, the first set of Euler-Lagrange equations, (\ref{vweEL1}), coincide with the  sigma-model governing harmonic maps of pseudo-Euclidean plane into a pseudo-Riemannian surface of constant curvature 1, see Appendix \ref{app-sigma}.

\begin{remark}
The Lagrangian multiform \eqref{LMF} possesses a translationally non-invariant version,
\begin{equation}\label{LMF1}
	\cL = \sum_{i,j}(x^i-x^j)\frac{u_{ij}^2}{u_iu_j}\ \d x^i\wedge \d x^j,
\end{equation}
with the corresponding variational equations
\begin{equation}\label{Vx}
	(x^i-x^j)u_{k}u_{ij}+(x^k-x^i)u_{j}u_{ik}+(x^j-x^k)u_{i}u_{jk}=0
\end{equation}
and \eqref{E}, which remains unchanged. Note that translationally non-invariant Veronese web equation \eqref{Vx} has appeared in \cite{kruglikov2017veronese, lobb2009lagrangian}.
\end{remark}

That the system $A_{ijk} = 0$, $B_{ijk} = 0$ implies the multiform Euler-Lagrange equations can be proved in general. (The various elements of the proof below can be found in \cite{petrera2021variational,sleigh2020variational,sleigh2023lagrangian,suris2016lagrangian}.)

\begin{theorem}
	\label{thm-double-zero}
	If the differential $\d \cL$ factorises as in equation \eqref{factorisation}, then the full system of multiform Euler-Lagrange equations \eqref{mfEL} follows from the system of equations $A_{ijk}=0, B_{ijk}=0.$
\end{theorem}
\begin{proof}
Consider an arbitrary surface $\Gamma$ and a variation $v$ of $u$, leading to a variation of the action
\[ \delta S_\Gamma =  \frac{\d}{\d \varepsilon} \bigg|_{\varepsilon=0} \int_\Gamma \cL[u + \varepsilon v] . \]
Without loss of generality we can assume that $v$ is supported on a small neighbourhood $\Omega \in \mathbb{R}^n$. Then $\delta S_\Gamma = \delta S_{\Gamma \cap \Omega}$ and we can find a three-dimensional volume $V$ such that inside of $\Omega$, the boundary of $V$ coincides with $\Gamma$, i.e.\@ $\Gamma \cap \Omega = \partial V \cap \Omega$. Then
\begin{align*}
	\delta S_\Gamma &= \delta S_{\partial V \cap \Omega} = \delta S_{\partial V} 
	= \frac{\d}{\d \varepsilon} \bigg|_{\varepsilon=0} \int_V \d \cL[u + \varepsilon v] \\
	&= \frac{\d}{\d \varepsilon} \bigg|_{\varepsilon=0} \int_V \sum_{i < j < k} A_{ijk}[u + \varepsilon v]B_{ijk}[u + \varepsilon v] \, \d x^i\wedge \d x^j\wedge \d x^k \\
	&= \int_V \sum_{i < j < k} \left(A_{ijk}[u] \frac{\d B_{ijk}[u + \varepsilon v]}{\d \varepsilon} + B_{ijk}[u] \frac{\d A_{ijk}[u + \varepsilon v]}{\d \varepsilon} \right)\! \bigg|_{\varepsilon=0} \d x^i\wedge \d x^j\wedge \d x^k,
\end{align*}
which is zero on solutions of the system $A_{ijk}=0, B_{ijk}=0$. Since $\Gamma$ was chosen arbitrarily, this implies that the multiform Euler-Lagrange equations hold.
\end{proof}

In the situation of Theorem \ref{thm-double-zero}, it is said that $\d \cL$ attains a \emph{double zero} on solutions of the system $A_{ijk}=0,\, B_{ijk}=0$. In other examples, the coefficients of $\d \cL$ are sums of factorised terms,
\[ d \cL =\sum_{i<j<k} \left( \sum_\ell A_{ijk}^\ell B_{ijk}^\ell \right) \, \d x^i\wedge \d x^j\wedge \d x^k , \]
so that $\d \cL$ attains a double zero on the system $A_{ijk}^\ell = 0,\, B_{ijk}^\ell = 0$. Again, this system implies the multiform Euler-Lagrange equations.

Factorisations of $\d \cL$ were first used in a narrower sense in \cite{petrera2021variational} and \cite{sleigh2020variational}. Since then, arguments involving double zeros have inspired the construction of Lagrangian multiforms in several different settings \cite{sleigh2022semidiscrete, caudrelier2024lagrangian,nijhoff2023lagrangian,richardson2025discrete,richardson2025discretea}.

The Lagrangian multiform (\ref{LMF}) has two important properties that motivated the structure of the present paper:

\begin{itemize}
	\item The relation $\d \cL=0 \ (\text{mod }A_{ijk}=0)$ suggests that $\cL$ is a (second-order) conservation law of the Veronese web hierarchy (\ref{V}), for which $B_{ijk}$ can be viewed as its characteristic. In Section \ref{sec2} we provide further examples of this kind  by using second-order conservation laws of some linearly degenerate PDEs in 3D such as the Veronese web hierarchy and the Mikhalev equation. It is yet to be explored to what extent conservation laws of integrable PDEs/hierarchies can be viewed as a common source of Lagrangian multiforms.
	
	\item It will be shown in Section \ref{sec3} that equations \eqref{V}--\eqref{E} coincide with the Gibbons-Tsarev equations governing hydrodynamic reductions of the 4D second heavenly equation. Further examples (such as  the first heavenly equation and a 6D version of the second heavenly equation) suggest that Gibbons-Tsarev type equations governing hydrodynamic reductions of linearly degenerate dispersionless integrable PDEs possess a Lagrangian multiform representation. Such Gibbons-Tsarev equations are only of interest in dimensions $d\geq 4$ as, according to \cite{odesskii2009systems}, Gibbons-Tsarev equations of three-dimensional linearly degenerate PDEs are linearisable. In contrast, Gibbons-Tsarev equations of higher-dimensional linearly degenerate PDEs are not linearisable (not even Darboux integrable).	
	 
\end{itemize}

This paper is of an experimental nature, in the sense that it presents a number of intriguing examples. They support two observations that relate to the points above and appear to be new:
(1) that Lagrangian 2-forms appear as (higher) conservation laws of various 3D integrable PDEs;
(2) that Lagrangian multiforms appear in the context of  Gibbons-Tsarev equations governing hydrodynamic reductions of various heavenly-type PDEs in 4D.

\section{Lagrangian multiforms and conservation laws of linearly degenerate PDEs}
\label{sec2}

The examples below support a point of view of Lagrangian multiforms as (higher) conservation laws of integrable PDEs (hierarchies).	Each of these examples starts from a 3D PDE with a known conservation law, without assuming any Lagrangian structure a priori. This is in contrast to previous works \cite{petrera2017variational,petrera2021variational,sleigh2020variational}, where Lagrangian multiforms are constructed by starting from a known Lagrangian and its variational symmetries.

\subsection{Veronese web hierarchy }

The Veronese web hierarchy (\ref{V}) possesses a second-order conservation law,
$$
\sum_{i<j}(c^i-c^j)\frac{u_{ij}^2}{u_iu_j}\, \d x^i\wedge \d x^j,
$$
as well as $n+1$ first-order conservation laws,
$$
\sum_{j\ne i}\frac{1}{c^i-c^j}\frac{u_j}{u_i}\, \d x^i\wedge \d x^j, \qquad \sum_{i< j}(c^i-c^j)u_iu_j\, \d x^i\wedge \d x^j ;
$$
in the first of these expressions, the index $i\in \{1, \dots, n\}$ is fixed, while $j$ varies (here under the `order' of a conservation law we understand the highest order of partial derivatives of $u$ involved). Taking their linear combination gives Lagrangian multiform
$$
\cL = \sum_{i < j} L_{ij}\, \d x^i\wedge \d x^j, \qquad L_{ij}=(c^i-c^j)\frac{u_{ij}^2}{u_iu_j}+\frac{1}{c^i-c^j}\left(n^j\frac{u_i}{u_j}+n^i\frac{u_j}{u_i} \right)+\varepsilon (c^i-c^j)u_iu_j,
$$
which can be seen as a multi-parameter deformation of the Lagrangian multiform (\ref{LMF}); here the constants $n^i$ and $ \varepsilon$ are the deformation parameters. The exterior derivative factorises as in equation \eqref{factorisation}, with
\[ A_{ijk} = (c^i-c^j)u_{k}u_{ij}+(c^j-c^k)u_{i}u_{jk}+(c^k-c^i)u_{j}u_{ik}=0 \]
and
\begin{align*}
	B_{ijk} &= 2 u_{ijk} - \left(\frac{u_{ij}u_{ik}}{u_i}+ \frac{u_{ij}u_{jk}}{u_j}+\frac{u_{ik}u_{jk}}{u_k}\right)\\
	&\quad + \frac{n^i u_ju_k}{(c^i-c^j)(c^i-c^k)u_i} +  \frac{n^j u_iu_k}{(c^j-c^i)(c^j-c^k)u_j} + \frac{n^k u_iu_j}{(c^k-c^i)(c^k-c^j)u_k} - \varepsilon u_iu_ju_k.
\end{align*}
Hence, by Theorem \ref{thm-double-zero}, the system $A_{ijk} = 0$, $B_{ijk} = 0$ implies the  multiform Euler-Lagrange equations, which are (compare with \eqref{vweEL}):
\begin{align*}
	&\!\left(\frac{u_{ij}^2}{u_i^2u_j} +\frac{1}{c^i-c^j}\left(\frac{n^j}{u_j} - \frac{n^i u_j}{u_i^2} \right)+\varepsilon (c^i-c^j)u_j \right)_i \\
	&\qquad + \left(\frac{u_{ij}^2}{u_iu_j^2} +\frac{1}{c^i-c^j}\left( -\frac{n^j u_i}{u_j^2} + \frac{n^i}{u_i} \right) + \varepsilon (c^i-c^j)u_i \right)_j + \left(\frac{2u_{ij}}{u_iu_j}\right)_{ij}=0, \\
	&(c^j - c^i) \left( \frac{u_{ij}^2}{u_i u_j^2} + 2 \frac{u_{ii}u_{ij}}{u_i^2 u_j} - 2 \frac{u_{iij}}{u_i u_j} - \frac{n^i}{u_i} - \frac{n^j u_i}{u_j^2} - \frac{n^k u_i}{u_k^2} + \varepsilon u_i \right) \\
	&\qquad - (c^k - c^i) \left( \frac{u_{ik}^2}{u_i u_k^2} + 2 \frac{u_{ii}u_{ik}}{u_i^2 u_k} - 2 \frac{u_{iik}}{u_i u_k} - \frac{n^i}{u_i} - \frac{n^j u_i}{u_j^2} - \frac{n^k u_i}{u_k^2} + \varepsilon u_i \right) = 0,  
	\\
	& (c^i-c^j)\frac{u_{ij}}{u_iu_j}+(c^j-c^k)\frac{u_{jk}}{u_ju_k}+(c^k-c^i)\frac{u_{ik}}{u_iu_k} = 0 .
\end{align*}
We emphasise that the system  $A_{ijk} = 0$, $B_{ijk} = 0$ is involutive, with the general solution depending on $2n$ arbitrary functions of one variable.

\subsection{Translationally non-invariant Veronese web hierarchy}

The translationally non-invariant Veronese web hierarchy (\ref{Vx}) also possesses a second-order conservation law,
$$
\sum_{i,j}(x^i-x^j)\frac{u_{ij}^2}{u_iu_j}\, \d x^i\wedge \d x^j,
$$
as well as $n+1$ first-order conservation laws,
$$
\sum_{j\ne i}\frac{1}{x^i-x^j}\frac{u_j}{u_i}\, \d x^i\wedge \d x^j, \qquad \sum_{i, j}(x^i-x^j)u_iu_j\, \d x^i\wedge \d x^j;
$$
in the first of these expressions, $i\in \{1, \dots, n\}$ is fixed, while $j$ varies. Taking their linear combination gives the Lagrangian multiform,
$$
{\cal L}=\sum L_{ij}\, \d x^i\wedge \d x^j, \qquad L_{ij}=(x^i-x^j)\frac{u_{ij}^2}{u_iu_j}+\frac{1}{x^i-x^j}\left(n^j\frac{u_i}{u_j}+n^i\frac{u_j}{u_i} \right)+\varepsilon (x^i-x^j)u_iu_j,
$$
which can be seen as a multi-parameter deformation of the Lagrangian multiform (\ref{LMF1}). This Lagrangian multiform, with $\varepsilon = 0$, has appeared in \cite{lobb2009lagrangian}. The corresponding variational equations $A_{ijk}=0$ and $B_{ijk}=0$ are
\begin{equation*}
(x^i-x^j)u_{k}u_{ij}+(x^j-x^k)u_{i}u_{jk}+(x^k-x^i)u_{j}u_{ik}=0
\end{equation*}
and
\begin{align*}
&u_{ijk} =\frac{1}{2}\left(\frac{u_{ij}u_{ik}}{u_i}+ \frac{u_{ij}u_{jk}}{u_j}+\frac{u_{ik}u_{jk}}{u_k}\right)\\
& - \frac{n^i u_ju_k}{2(x^i-x^j)(x^i-x^k)u_i} - \frac{n^j u_iu_k}{2(x^j-x^i)(x^j-x^k)u_j} -  \frac{n^k u_iu_j}{2(x^k-x^i)(x^k-x^j)u_k} + \varepsilon \frac{u_iu_ju_k}{2},
\end{align*}
respectively. As in the previous case,  the system  $A_{ijk} = 0$, $B_{ijk} = 0$ is involutive, with the general solution depending on $2n$ arbitrary functions of one variable.

\subsection{Mikhalev equation}

The Mikhalev equation has the form
\begin{equation}\label{Mik}
u_{33}-u_{12}+u_3u_{11}-u_1u_{13}=0.
\end{equation}
It first appeared in \cite{mikhalev1992hamiltonian}, 
in the context of the Hamiltonian formalism of KdV type hierarchies. It possesses four first-order and seven second-order conservation laws (see \cite{baran2014higher} for the characteristics of the latter), 
$$
\cL =F\, \d x^2\wedge \d x^3+G\, \d x^3\wedge \d x^1+H\, \d x^1\wedge \d x^2.
$$ 
For any such conservation law (or a linear combination thereof), the differential $\d\cL=(F_1+G_2+H_3)\, \d x^1\wedge \d x^2\wedge \d x^3$  factorises as 
$$
F_1+G_2+H_3=(u_{33}-u_{12}+u_3u_{11}-u_1u_{13})\cdot (\Sigma),
$$
where the characteristic $\Sigma$ is a differential expression in $u$ such that the combined system $\{u_{33}-u_{12}+u_3u_{11}-u_1u_{13}=0,\ \Sigma=0\}$ is involutive. Thus, these conservation laws give rise to Lagrangian multiforms, where $L_{12} = H$, $L_{13} = -G$, and $L_{23} = F$.
Below we discuss some of the simplest examples of this constuction.

\paragraph{Case 1.} Let
\begin{align*}
	&F=u_3u_{11}^2-u_{13}^2-2u_{11}(u_{33}-u_{12}+u_3u_{11}-u_1u_{13}),\\
	&G=-u_{11}^2, \\
	&H=2u_{11}u_{13}-u_1u_{11}^2.
\end{align*}
Due to the factorisation
$$
F_1+G_2+H_3=-2(u_{33}-u_{12}+u_3u_{11}-u_1u_{13})\, u_{111},
$$
as well as the involutivity of  the combined system,
\begin{equation}
	\label{mikhalev-1AB}
	u_{33}-u_{12}+u_3u_{11}-u_1u_{13}=0, \qquad u_{111}=0,
\end{equation}
$\cL$ has all properties of a Lagrangian multiform. The multiform Euler-Lagrange equations \eqref{mfEL} in this case are all immediate consequences of the system \eqref{mikhalev-1AB}. In particular, we find the two equations of \eqref{mikhalev-1AB} as multiform Euler-Lagrange equations of type \eqref{mfEL-ukl} and \eqref{mfELb}:
\begin{align*}
	& \var{23}{F}{u_{11}} = \der{F}{u_{11}} = -2(u_{33}-u_{12}+ u_3 u_{11}-u_1u_{13}) , \\
	& \var{31}{G}{u_1} + \var{23}{F}{u_2} = -\partial_1 \der{G}{u_{11}} + 0 = 2 u_{111} .
\end{align*}

This example can be deformed by adding to $\cL$ a first-order conservation law of the Mikhalev equation, 
$$
\left( 2u_1u_3^2-u_1^3u_3-u_2u_3 \right)_1 + \left( u_1^3-u_1u_3 \right)_2 + \left( u_1^4-3u_1^2u_3+u_1u_2+u_3^2 \right)_3=0.
$$
Thus, we take
\begin{align*}
&\tilde F=u_3u_{11}^2-u_{13}^2-2u_{11}(u_{33}-u_{12}+u_3u_{11}-u_1u_{13})+2u_1u_3^2-u_1^3u_3-u_2u_3,\\
&\tilde G=-u_{11}^2+u_1^3-u_1u_3,\\
&\tilde H=2u_{11}u_{13}-u_1u_{11}^2+u_1^4-3u_1^2u_3+u_1u_2+u_3^2,
\end{align*}
with the factorisation 
$$
\tilde F_1+\tilde G_2+\tilde H_3= 2 \left(u_{33}-u_{12}+u_3u_{11}-u_1u_{13} \right) \left( u_3-u_{111}-\frac{3}{2}u_1^2 \right).
$$
Due to involutivity of the combined system,
\begin{equation}
	\label{mikhalev-1-deformed-AB}
u_{33}-u_{12}+u_3u_{11}-u_1u_{13}=0, \qquad u_3-u_{111}-\frac{3}{2}u_1^2=0,
\end{equation}
the deformed $\tilde \cL = \tilde F\, \d x^2\wedge \d x^3 + \tilde G\, \d x^3\wedge \d x^1 + \tilde H\, \d x^1\wedge \d x^2$ is also a Lagrangian multiform. Note that the second equation is the potential KdV equation.

Again, we find the two equations of \eqref{mikhalev-1-deformed-AB} as multiform Euler-Lagrange equations of type \eqref{mfEL-ukl} and \eqref{mfELb}:
\begin{align*}
	& \var{23}{\tilde F}{u_{11}} = \der{\tilde F}{u_{11}} = -2(u_{33}-u_{12}+ u_3 u_{11}-u_1u_{13}) , \\
	& \var{31}{\tilde G}{u_1} + \var{23}{\tilde F}{u_2} = \der{\tilde G}{u_{1}} - \partial_1 \der{\tilde G}{u_{11}} + \der{\tilde F}{u_2} 
	= (3 u_1^2 - u_3) + 2 u_{111} - u_3,
\end{align*}
and all other multiform Euler-Lagrange equations are consequences of these.

\paragraph{Case 2.} Let
\begin{align*}
& F=(2u_3u_{11}-u_{12})u_{13}-u_1u_3u_{11}^2+2(u_1u_{11}-u_{13})(u_{33}-u_{12}+u_3u_{11}-u_1u_{13}), \\
& G=u_1u_{11}^2-u_{11}u_{13}, \\ & H=u_1^2u_{11}^2-2u_1u_{11}u_{13}-u_3u_{11}^2+u_{11}u_{12}+u_{13}^2.
\end{align*}
Due to the factorisation
$$
F_1+G_2+H_3 = 2 \left(u_{33}-u_{12}+u_3u_{11}-u_1u_{13} \right) \left( u_1u_{111}-u_{113}+\frac{1}{2}u_{11}^2 \right),
$$
as well as the involutivity of  the combined system,
$$
u_{33}-u_{12}+u_3u_{11}-u_1u_{13}=0, \qquad u_1u_{111}-u_{113}+\frac{1}{2}u_{11}^2=0,
$$
$\cL$ is a Lagrangian multiform. Note that the second equation is equivalent to the Hunter-Saxton equation, 
\[ (v_3-vv_1)_1+\frac{1}{2}v_1^2=0,\]
for $v=u_1$, which is an integrable PDE that arises in the theory of nematic liquid crystals \cite{hunter1991dynamics}.
This equation appears as a multiform Euler-Lagrange equation of type \eqref{mfELb}:
\[ \var{13}{G}{u_1} - \var{23}{F}{u_2} = -u_{11}^2 - 2 u_1 u_{111} + 2 u_{113}
	= 2\left( (u_{13}-u_1u_{11})_1+\frac{1}{2}u_{11}^2 \right).
\]

\paragraph{Case 3.} Let
\begin{align*}
&F=u_3(u_3-u_1^2)u_{11}^2+2(u_3-u_1^2)(u_{33}-u_{12}-u_1u_{13})u_{11}
\\&\qquad -2u_1^2u_{13}^2+2u_1u_{13}u_{33}+u_3u_{13}^2+u_{12}(u_{12}-2u_{33}), \\
&G=-(u_1u_{11}-u_{13})^2, \\
&H=(u_1u_{11}-u_{13})(-u_1^2u_{11}+u_1u_{13}+2u_3u_{11}-2u_{12}).
\end{align*}
Due to the factorisation
$$
F_1+G_2+H_3=-2(u_{33}-u_{12}+u_3u_{11}-u_1u_{13})(u_{112}+(u_1^2-u_3)u_{111}-u_1u_{113}+u_1u_{11}^2-u_{11}u_{13}),
$$
as well as the involutivity of  the combined system,
$$
u_{33}-u_{12}+u_3u_{11}-u_1u_{13}=0, \quad u_{112}+(u_1^2-u_3)u_{111}-u_1u_{113}+u_1u_{11}^2-u_{11}u_{13}=0,
$$
$\cL$ is a Lagrangian multiform. Note that, modulo (\ref{Mik}),  the second equation is equivalent to the Gurevich-Zybin equation, 
\[ (\partial_3-v\, \partial_1)^2v=0, \]
for $v=u_1$, which is known to be linearisable by a reciprocal transformation, see 
\cite{gurevich1988nondissipative,gurevich1995largescale,pavlov2005gurevich}.
This equation appears as a multiform Euler-Lagrange equation of type \eqref{mfELb}:
\[ \var{13}{G}{u_1} - \var{23}{F}{u_2} = 2 u_{1} u_{11}^{2} + 2 u_{1}^{2} u_{111} - 4 u_{1} u_{113} - 2 u_{11} u_{13} + 2 u_{133}
= 2 (\partial_3-u_1 \, \partial_1)^2 u_1 .
\]

An equivalent multiform can be obtained by subtracting $(u_{33} - u_1u_{13} + u_{11}u_3 - u_{12})^2$ from $F$:
\begin{align*}
\tilde F &= 2 u_{1}^{3} u_{11} u_{13} - u_{1}^{2} u_{11}^{2} u_{3} + 2 u_{1}^{2} u_{11} u_{12} - 3 u_{1}^{2} u_{13}^{2} - 2 u_{1}^{2} u_{11} u_{33} - 2 u_{1} u_{12} u_{13} + u_{13}^{2} u_{3} \\
	&\qquad + 4 u_{1} u_{13} u_{33} - u_{33}^{2},\\
	\tilde G &= - \left( u_{1}u_{11} - u_{13} \right)^2, \\
	\tilde H &= -u_{1}^{3} u_{11}^{2} + 2 u_{1}^{2} u_{11} u_{13} + 2 u_{1} u_{11}^{2} u_{3} - 2 u_{1} u_{11} u_{12} - u_{1} u_{13}^{2} - 2 u_{11} u_{13} u_{3} + 2 \, u_{12} u_{13}.
\end{align*}
Its exterior derivative factorises as
\[ \tilde F_1 + \tilde G_2 + \tilde H_3 = -2(u_{33}-u_{12}+u_3u_{11}-u_1u_{13}) (u_{1} u_{11}^{2} + u_{1}^{2} u_{111} - 2 u_{1} u_{113} - u_{11} u_{13} + u_{133} ) ,
\]
where the second factor is exactly the Gurevich-Zybin equation.

Similar second-order conservation laws, as well as the associated Lagrangian multiforms, can be constructed for other 3D linearly degenerate second-order dispersionless integrable PDE as classified in \cite{ferapontov2015linearly}. We refer to \cite{baran2014higher, Morozov2022higher} where the characteristics of these conservation laws were calculated using third-order symmetries of the cotangent coverings of these equations.

\section{Lagrangian multiforms and Gibbons-Tsarev equations }
\label{sec3}

In this section we present Lagrangian multiforms for some Gibbons-Tsarev equations. To introduce Gibbons-Tsarev equations, we begin with a brief review of the method of hydrodynamic reductions.
In the most general setting, the method of hydrodynamic reductions applies to quasilinear PDEs of the form
\begin{equation}
\sum_{i=1}^d A^i(u)\, u_{x^i}=0,
\label{1}
\end{equation}
or any other classes of PDEs transformable into  form (\ref{1}), see below. Here $u$ is a (column) vector-function of $d$ independent variables $x^1, \dots, x^d$, and $u_{x^i}$ denote partial derivatives (the matrices $A^i(u)$ do not need to be square).
For definiteness, let us consider $4$-dimensional PDEs ($d=4$) with four independent variables  $t, x, y, z$. 
Let us look for  solutions of  (\ref{1}) in the form
${ u}={ u}(R)$  where the \emph{Riemann invariants}
$R=\{R^1, ..., R^n\}$ solve a triple of
commuting diagonal systems
\begin{equation}
R^i_t=\lambda^i(R)\, R^i_x, \qquad R^i_y=\mu^i(R)\, R^i_x, \qquad R^i_z=\eta^i(R)\,R^i_x.
\label{R}
\end{equation}
Note that the number $n$ of Riemann invariants is allowed to be
arbitrary.  Thus, the original  multi-dimensional equation (\ref{1}) is
decoupled into a collection of commuting  $(1+1)$-dimensional systems in Riemann invariants. 
Solutions of this type, known as nonlinear interactions of $n$ planar
simple waves, were  investigated in gas dynamics and
magnetohydrodynamics  \cite{burnat1970method, peradzynski1971nonlinear}. Later on, they 
reappeared in the context of the dispersionless KP hierarchy
\cite{kodama1989method, gibbons1996reductions, gibbons1999conformal}.

We recall, see \cite{tsarev1991geometry}, that the requirement of  
commutativity of the flows (\ref{R})
is equivalent to the following restrictions on their \emph{characteristic speeds}:
\begin{equation}
\frac{\lambda_j
^i}{\lambda^j-\lambda^i}=\frac{\mu_j^i}{\mu^j-\mu^i}=\frac{\eta_j^i}{\eta^j-\eta^i},
\label{comm}
\end{equation}
$i\ne j, \  \lambda_j^i=\frac{\partial}{\partial { R^j}}\lambda^i$, etc.  Substituting ${ u}(R)$ into (\ref{1}) and using
(\ref{R}), one  arrives at an over-determined system of equations for  $u(R)$ and the characteristic speeds $\lambda^i(R), \mu^i(R), \eta^i(R)$, known as  Gibbons-Tsarev equations.  These equations imply, in particular, that the characteristic speeds $\lambda^i, \mu^i$ and $\eta^i$ must satisfy an algebraic relation which can be interpreted 
as the dispersion relation of system (\ref{1}).

 One can show that the maximum ``amount''  of $n$-component reductions a $d$-dimensional PDE 
may possess is parametrised, modulo changes of variables $R^i\to
f^i(R^i)$,  by $(d-2)n$ arbitrary functions of one variable. 

\begin{definition}[\cite{ferapontov2004integrability, ferapontov2004hydrodynamic}]  A $d$-dimensional system (\ref{1}) is said to be integrable if its $n$-component reductions are locally parametrised by $(d-2)n$ arbitrary functions of one variable.
\end{definition}

Although integrability in the sense of hydrodynamic reductions is somewhat different from the familiar solitonic integrability, is has all attributes of any ``reasonable'' definition of integrability: 
\begin{itemize}

\item it is based on exact solutions (coming from the method of hydrodynamic reductions) which are locally dense in the space of all solutions of a given PDE; these solutions can be considered as natural `dispersionless analogues' of  multisoliton solutions of conventional integrable PDEs;

\item it is algorithmically verifiable;

\item it leads to classification results of  integrable PDEs within various particularly interesting classes;

\item it can be reformulated geometrically  as a certain involutivity condition of the principal symbol of the given PDE system.

\end{itemize}

All these properties are thoroughly discussed in \cite{ferapontov2004integrability, ferapontov2004hydrodynamic, ferapontov2014dispersionless, berjawi2020secondorder, berjawi2022secondorder}. Our main observation is that many Gibbons-Tsarev type equations governing hydrodynamic reductions of various dispersionless integrable systems  possess a Lagrangian multiform representation, in all dimensions $d\geq 4$. Below we provide Lagrangian multiform representations of Gibbons-Tsarev equations 
for some well-known integrable PDEs.

\subsection{Second heavenly equation}

Pleba\'nski's second heavenly equation,
\begin{equation}
\theta_{tx}+\theta_{zy}+\theta_{xx}\theta_{yy}-\theta^2_{xy}=0,
\label{h1}
\end{equation}
describes self-dual Ricci-flat metrics of the form
$$
\d s^2 = \d x \, \d t + \d y \, \d z + \theta_{yy} \, \d x^2-2\theta_{xy} \, \d x \, \d y + \theta_{xx} \, \d y^2,
$$
see \cite{plebanski1975solutions}. Hydrodynamic reductions of equation (\ref{h1}) were discussed in  \cite{ferapontov2004integrability}. Introducing the notation $\theta_{xx}=u, $ $ \theta_{xy}=v, $ $ \theta_{yy}=w, $ $ \theta_{tx}=p, $ $ \theta_{zy}=v^2-uw-p$, one first rewrites 
(\ref{h1}) in quasilinear form (\ref{1}),
\begin{equation}
\begin{split}
&u_y=v_x, \qquad u_t=p_x, \qquad v_y=w_x, \qquad v_t=p_y, \\
&v_z=(v^2-uw-p)_x, \qquad w_z=(v^2-uw-p)_y.
\end{split}
\label{h2}
\end{equation}
Hydrodynamic reductions are sought in the form $u=u(R^1, ..., R^n),$ $v=v(R^1, ..., R^n),$ $w=w(R^1, ..., R^n),$ $p=p(R^1, ..., R^n)$  where the Riemann invariants $R^i$ satisfy equations \eqref{R}. For every solution of \eqref{R} we require that the corresponding $u,v,w,p$ solve \eqref{h2}. This implies 
\begin{equation}
 p_i=\lambda^i u_i, \qquad  v_i=\mu^i u_i, \qquad w_i=(\mu^i)^2 u_i,
\label{h4}
\end{equation}
along with the dispersion relation
\begin{equation}
\lambda^i=2v\mu^i-w-u(\mu^i)^2-\mu^i\eta^i.
\label{h5}
\end{equation}
Recall that low indices in (\ref{h4}) indicate partial derivatives by the variables $R^i$. Substituting  $\lambda^i$ into the commutativity conditions (\ref{comm})
and taking into account that the compatibility conditions for the relations $ p_i=\lambda^i u_i, \  v_i=\mu^i u_i$ imply
$$
u_{ij}=\frac{\mu_j^i}{\mu^j-\mu^i}u_i+\frac{\mu_i^j}{\mu^i-\mu^j}u_j.
$$
One arrives at the following Gibbons-Tsarev type equations:
\begin{equation}
\begin{split}
& \mu_j^i=\frac{(\mu^j-\mu^i)^2}{\eta^j-\eta^i+u(\mu^j-\mu^i)} \ u_j, \\
& \eta^i_j=\frac{(\mu^j-\mu^i)(\eta^j-\eta^i)}{\eta^j-\eta^i+u(\mu^j-\mu^i)} \ u_j, \\
& u_{ij}=2\frac{\mu^j-\mu^i}{\eta^j-\eta^i+u(\mu^j-\mu^i)} \ u_iu_j.
\end{split}
\label{h6}
\end{equation}
Solving  equations (\ref{h6}) for $\mu^i,  \eta^i$ and $u$, determining $\lambda^i$ from (\ref{h5}) and calculating $p, v, w$ from  equations (\ref{h4}) 
(which are automatically compatible by virtue of (\ref{h6})), one obtains a general $n$-component hydrodynamic reduction of the second heavenly equation.
Moreover, the commutativity conditions  will also be satisfied identically.

We emphasize that  system (\ref{h6}) is in involution and its general solution depends on $3n$ arbitrary functions of one variable.
Indeed, one can arbitrarily prescribe the restrictions of $\mu^i$ and $\eta^i$ to the $R^i$-coordinate line. This  gives $2n$ arbitrary functions. Moreover, one can arbitrarily 
prescribe the restriction of $u$ to each of the coordinate lines, which provides extra $n$  arbitrary functions. Since  reparametrizations 
$R^i\to f^i(R^i)$ leave the system (\ref{h6}) invariant, one concludes that  general $n$-component reductions are locally parametrized by $2n$ arbitrary functions of one variable. This supports the evidence that the heavenly equation (\ref{h1}) is a  four-dimensional integrable PDE \cite{ferapontov2004integrability}.

Let us proceed with the analysis of the Gibbons-Tsarev system (\ref{h6}). Introducing $c^i=\eta^i+u\mu^i-v$, one readily obtains $c^i_j=0$ so that $c^i=c^i(R^i)$ are arbitrary functions of the indicated variables. Ultimately, system (\ref{h6}) simplifies to
\begin{subequations}
\label{h7}
\begin{align}
& \mu^i_j=\frac{(\mu^j-\mu^i)^2}{c^j-c^i} \ u_j, \label{h7a}\\
& u_{ij}=2\frac{\mu^j-\mu^i}{c^j-c^i} \ u_iu_j. \label{h7b}
\end{align}
\end{subequations}
The elimination of $\mu$'s from system (\ref{h7}) leads  to  equations for the variable $u$ alone. These can be obtained as follows. The equations (\ref{h7b}) lead to second-order PDEs for $u$ that form a translationally non-invariant deformation of the Veronese web hierarchy,
\begin{equation*}
(c^j-c^i)u_ku_{ij}+(c^i-c^k)u_ju_{ik}+(c^k-c^j)u_iu_{jk}=0,
\end{equation*}
see  \cite{zakharevich2000nonlinear, krynski2016paraconformal, lobb2009lagrangian, kruglikov2017veronese}.  Differentiating the equations \eqref{h7b} by $R^k$, one obtains a collection of third-order PDEs (\ref{E}) for $u$,
\begin{equation*}
u_{ijk}=\frac{1}{2}\left(\frac{u_{ij}u_{ik}}{u_i} +\frac{u_{ij}u_{jk}}{u_j}+\frac{u_{ik}u_{jk}}{u_k}\right).
\end{equation*}
Upon the identification $R^i\leftrightarrow x^i$, these equations correspond to Lagrangian multiform (\ref{LMF}),
\begin{equation*} 
{\cal L}=\sum_{i,j}(c^i-c^j)\frac{u_{ij}^2}{u_iu_j}\ \d R^i\wedge \d R^j,
\end{equation*}
where $c^i(R^i)$ are arbitrary functions of the indicated variables (the choices $c^i=R^i$ and $c^i=const$ are of particular interest).

\paragraph{Lagrangian multiform for system \eqref{h7}.} System \eqref{h7} is related to the Lagrangian multiform
\begin{equation}\label{LMFnew}
	\cL = \sum_{i,j} \left( \mu^j_i u_j - \mu^i_j u_i + \frac{(\mu^j - \mu^i)^2}{c^j - c^i} u_i u_j \right) \d R^i\wedge \d R^j .
\end{equation}
Its multiform Euler-Lagrange equations are as follows:
\begin{itemize}
	\item \eqref{mfELb} yields
	\[ \mu_j^i + \frac{(\mu^i - \mu^j)^2}{c^i- c^j} u_j = \mu_j^k + \frac{(\mu^k - \mu^j)^2}{c^k - c^j} u_j. \]
	In other words, this equation implies that the quantity
	\[ P_j:= \mu_j^i+ \frac{(\mu^i- \mu^j)^2}{c^i - c^j} u_j \]
	does not depend on the choice of $i$.

	\item  \eqref{mfELa} yields
	\[ \der{P_i}{R^j} = \der{P_j}{R^i}. \]
	It follows that there exists some function $P$ such that $P_i = \der{}{R^i} P$ for all $i$.
	
	\item \eqref{mfELc} is trivially satisfied: all three variational derivatives are identically zero.
	
	\item \eqref{mfELa} with respect to $\mu^i$ instead of $u$ yields
	\[ -u_{ij} + 2 \frac{\mu^j - \mu^i}{c^j - c^i} u_i u_j = 0 .\]
	
	\item \eqref{mfELb} with respect to $\mu^i$ instead of $u$ yields the trivial equation $-u_i = -u_i$.
	
	\item \eqref{mfELc} with respect to $\mu^i$ instead of $u$ is trivially satisfied: all three variational derivatives are identically zero.
\end{itemize}
Hence, the system of multiform Euler-Lagrange equations is equivalent to
\begin{align*}
	& \mu_j^i = \frac{(\mu^j - \mu^i)^2}{c^j - c^i} u_i + \der{P}{R^j}, \\
	&  u_{ij} = 2 \frac{\mu^j - \mu^i}{c^j - c^i} u_i u_j,
\end{align*}
which, except for the term $\der{P}{R^J}$, matches equation \eqref{h7}. Note that we can remove this term by a gauge transformation $\tilde \mu^i = \mu^i - P$.
For the Lagrangian multiform (\ref{LMFnew}),  each coefficients of $\d \cL = \sum_{ijk} M_{ijk} \,\d R^i \wedge \d R^j \wedge \d R^k$ decomposes into a sum of three factorised terms:
\begin{align*}
	M_{ijk} &= \left( u_{ij} - 2 \frac{\mu^j-\mu^i}{c^j-c^i} u_i u_j \right) \left( \mu^i_k - \frac{(\mu^k - \mu^i)^2}{c^k - c^i}u_k - \mu^j_k + \frac{(\mu^k - \mu^j)^2}{c^k - c^j}u_k \right) \\
	&\qquad + \left( u_{jk} - 2 \frac{\mu^k-\mu^j}{c^k-c^j} u_j u_k \right) \left( \mu^j_i - \frac{(\mu^i - \mu^j)^2}{c^i - c^j}u_i - \mu^k_i + \frac{(\mu^i - \mu^k)^2}{c^i - c^k}u_i \right) \\
	&\qquad + \left( u_{ik} - 2 \frac{\mu^i-\mu^k}{c^i-c^k} u_k u_i \right) \left( \mu^k_j - \frac{(\mu^j - \mu^k)^2}{c^j - c^k}u_j - \mu^i_j + \frac{(\mu^j - \mu^i)^2}{c^j - c^i}u_j \right).
\end{align*}

\subsection{First heavenly equation}

Pleba\'nski's first heavenly equation,
\begin{equation}
\Omega_{xy}\Omega_{zt}-\Omega_{xt}\Omega_{zy}=1,
\label{P1}
\end{equation}
governs K\"ahler potentials of  4-dimensional self-dual Ricci-flat metrics
\begin{equation*}
\d s^2 = 2\Omega_{xy} \,\d x \, \d y + 2\Omega_{zt} \, \d z \,\d t + 2\Omega_{xt} \, \d x \, \d t + 2 \Omega_{zy} \, \d z \, \,\d y,
\label{metric0}
\end{equation*}
see \cite{plebanski1975solutions}. Hydrodynamic reductions of equation (\ref{P1}) were discussed in \cite{ferapontov2003hydrodynamic}.
Using the notation $\Omega_{xy}=a, \ \Omega_{zt}=b, \ \Omega_{xt}=p, \ \Omega_{zy}=q$,
equation (\ref{P1}) takes quasilinear form (\ref{1}),
\begin{equation}
a_t=p_y, \qquad a_z=q_x, \qquad b_x=p_z, \qquad b_y=q_t, \qquad ab-pq=1.
\label{P3}
\end{equation}
Hydrodynamic reductions are sought in the form
$a=a(R^1, ..., R^n)$, $b=b(R^1, ..., R^n)$, $ p=p(R^1, ..., R^n)$, $q=q(R^1, ..., R^n)$, 
where the Riemann invariants  $R^i$ solve a triple of commuting hydrodynamic type systems  (\ref{R}).
The substitution  into (\ref{P3}) implies
\begin{equation}
q_i=\eta^ia_i, \qquad p_i=\frac{\lambda^i}{\mu^i}a_i, \qquad 
b_i=\frac{\eta^i \lambda^i}{\mu^i}a_i.
\label{partial}
\end{equation}
Differentiation of $ab-pq=1$ implies $ab_i+ba_i=pq_i+qp_i$ which, by virtue of (\ref{partial}), gives the dispersion relation
\begin{equation}
\lambda^i=\mu^i\frac{p\eta^i-b}{a\eta^i-q}.
\label{v}
\end{equation}
Substituting (\ref{v}) back into (\ref{partial}), one obtains
\begin{equation}
q_i=\eta^ia_i, \qquad 
p_i=\frac{p\eta^i-b}{a\eta^i-q}a_i, \qquad 
b_i=\eta^i \frac{p\eta^i-b}{a\eta^i-q}a_i.
\label{pb}
\end{equation}
Calculation of the compatibility conditions for equations (\ref{pb}) results in
\begin{equation}
\begin{split}
& a_{ij}=\frac{\eta^i_j}{\eta^j-\eta^i}a_i+\frac{\eta^j_i}{\eta^i-\eta^j}a_j, \\
& \eta^i_j(a\eta^j-q)a_i+\eta^j_i(a\eta^i-q)a_j=
(\eta^i-\eta^j)^2a_ia_j.
\end{split}
\label{a}
\end{equation}
The commutativity conditions (\ref{comm}),
with $\lambda^i$ given by (\ref{v}),  imply 
\begin{equation}
\eta^i_j=\frac{\mu^i(\eta^j-\eta^i)^2}{\mu^i(a\eta^j-q)-\mu^j(a\eta^i-q)}a_j.
\label{c}
\end{equation}
Ultimately, combining (\ref{a}), (\ref{comm}) and (\ref{c}), we arrive at the following Gibbons-Tsarev type equations for $a, q, \mu^i, \eta^i$:
\begin{equation}
\begin{split}
& q_i=\eta^ia_i, \\
& a_{ij}=\frac{(\mu^i+\mu^j)(\eta^j-\eta^i)}{\mu^i(a\eta^j-q)-\mu^j(a\eta^i-q)} a_ia_j,  \\
& \eta^i_j=\frac{\mu^i(\eta^j-\eta^i)^2}{\mu^i(a\eta^j-q)-\mu^j(a\eta^i-q)} a_j, \\ &\mu^i_j=\frac{\mu^i(\eta^j-\eta^i)(\mu^j-\mu^i)}{\mu^i(a\eta^j-q)-\mu^j(a\eta^i-q)} a_j,
\label{nP2}
\end{split}
\end{equation}
where $i\ne j$. Note that $p, b, \lambda^i$ do not explicitly enter these equations. One can show by  direct calculation that  system (\ref{nP2}) is in involution and its general solution depends, modulo reparametrisations $R^i\to f^i(R^i)$, on $2n$ arbitrary functions of a single variable, thus confirming the integrability of the first heavenly equation.

Let us proceed with the analysis of the Gibbons-Tsarev system (\ref{nP2}).  Introducing potential $u$ by the formula
$u_i=a_i/\mu^i$ (compatibility is guaranteed by (\ref{nP2})), as well as the functions $c^i=(a\eta^i-q)/\mu^i$, one can rewrite (\ref{nP2}) in the form
\begin{equation}
\begin{split}
& q_i=\mu^i\eta^iu_i, \\
& u_{ij}=2\frac{\eta^j-\eta^i}{c^j-c^i}\, u_iu_j,  \\
& \eta^i_j=\frac{(\eta^j-\eta^i)^2}{c^j-c^i}\, u_j, \\
&\mu^i_j=\frac{(\eta^j-\eta^i)(\mu^j-\mu^i)}{c^j-c^i}\, u_j .
\label{nP22}
\end{split}
\end{equation}
It remains to  verify the identity $c^i_j=0$, i.e.\@ that the $c_i$ are functions of $R^i$ only, which makes the above equations for $u$ and $\eta^i$ identical to equations (\ref{h7}). Thus, Gibbons-Tsarev systems governing hydrodynamic reductions of the
first and second heavenly equations come from one and the same Lagrangian multiform (which is not surprising as both  equations describe one and the same class of self-dual Ricci-flat metrics and are B\"acklund-related \cite{plebanski1975solutions}).

\subsection{6D version of the second heavenly equation}

A six-dimensional generalisation of the second heavenly equation,
\begin{equation}
\theta_{t\tilde t}+\theta_{z\tilde z}+\theta_{tx}\theta_{zy}-\theta_{ty}\theta_{zx}=0,
\label{P11}
\end{equation}
was proposed in \cite{plebanski1996lagrangian}.
Its hydrodynamic reductions were studied in  \cite{ferapontov2004integrability}. Introducing  the notation $\theta_{tx}=a, $ $ \theta_{zy}=b, $ 
$ \theta_{ty}=p, $ $ \theta_{zx}=q, $ $\theta_{z\tilde z}=r,$ $ \theta_{t\tilde t}=pq-ab-r$, one can rewrite 
(\ref{P11}) in quasilinear form (\ref{1}),
\begin{equation}
\begin{split}
&a_y=p_x, \qquad a_z=q_t, \qquad b_t=p_z, \qquad b_x=q_y, \qquad b_{\tilde z}=r_y, \qquad q_{\tilde z}=r_x,\\
& p_{\tilde z}=(pq-ab-r)_y.
\end{split}
\label{P22}
\end{equation}
Hydrodynamic reductions are sought in the form $a=a(R^1, ..., R^n),$ $b=b(R^1, ..., R^n),$ $p=p(R^1, ..., R^n),$ $q=q(R^1, ..., R^n), $  $r=r(R^1, ..., R^n)$ where the Riemann invariants $R^1, ..., R^n$  solve the commuting equations
\begin{align*}
&R^i_{x}=\lambda^i(R)\ R^i_{z}, \qquad  R^i_{ y}=\mu^i(R)\ R^i_{z}, \qquad R^i_{\tilde z}=\eta^i(R)\ R^i_{z}, \qquad \\
& R^i_{t}=\beta^i(R)\ R^i_{z}, \qquad R^i_{\tilde t}=\gamma^i(R)\ R^i_{z}.
\end{align*}
The substitution into (\ref{P22}) implies
\begin{equation}
\partial_i p=\beta^i\partial_i b, \qquad \partial_i r=\frac{\eta^i}{\mu^i}\partial_i b, \qquad \partial_iq=\frac{\lambda^i}{\mu^i}\partial_i b, \qquad \partial_i a=\frac{\lambda^i \beta^i}{\mu^i}\partial_i b,
\label{P4}
\end{equation}
along with the dispersion relation
\begin{equation}
\eta^i=\beta^i\mu^i q+\lambda^i p-\beta^i\lambda^i b-\mu^i a-\beta^i\gamma^i.
\label{P5}
\end{equation}
Substituting  $\eta^i$ into the commutativity conditions
$$
\frac{\partial_j\lambda
^i}{\lambda^j-\lambda^i}=\frac{\partial_j\mu^i}{\mu^j-\mu^i}=\frac{\partial_j\eta^i}{\eta^j-\eta^i}=
\frac{\partial_j\beta^i}{\beta^j-\beta^i}=\frac{\partial_j\gamma^i}{\gamma^j-\gamma^i},
$$
and taking into account that the compatibility conditions for the relations $\partial_i p=\beta^i\partial_i b$ imply
$$
\partial_i\partial_jb=\frac{\partial_j\beta^i}{\beta^j-\beta^i}\partial_ib+\frac{\partial_i\beta^j}{\beta^i-\beta^j}\partial_jb,
$$
one arrives at the following system:
\begin{equation}
\begin{split}
& \frac{\partial_j\beta^i}{\beta^j-\beta^i}=\frac{\partial_j\lambda^i}{\lambda^j-\lambda^i}=
\frac{\partial_j\mu^i}{\mu^j-\mu^i}=\frac{\partial_j\gamma^i}{\gamma^j-\gamma^i}=
\frac{\lambda^i-\lambda^j\mu^i/\mu^j}{q(\mu^j-\mu^i)+b(\lambda^i-\lambda^j)+\gamma^i-\gamma^j}
\ \partial_jb, \\
& \partial_i\partial_jb=
\frac{\lambda^i(1+\mu^j/\mu^i)-\lambda^j(1+\mu^i/\mu^j)}{q(\mu^j-\mu^i)+b(\lambda^i-\lambda^j)+\gamma^i-\gamma^j}\ \partial_ib\partial_jb.
\end{split}
\label{P6}
\end{equation}
Solving   equations (\ref{P6}) for $\beta^i$, $\lambda^i$, $\mu^i$,  $\gamma^i$ and $b$, determining $\eta^i$ from the dispersion relation (\ref{P5}) and calculating $p, r, q, a$ from the equations (\ref{P4}) 
(which are automatically compatible by virtue of (\ref{P6})), one obtains a general $n$-component  reduction of the  equation (\ref{P11}).
The commutativity conditions  will  be satisfied identically.
System (\ref{P6}) is in involution and its general solution depends, up to
reparametrisations $R^i\to \varphi^i(R^i)$,  on $4n$ arbitrary functions of one variable.

Let us proceed with the analysis of Gibbons-Tsarev system (\ref{P6}).  Introducing the functions $c^i=\gamma^i+b\lambda^i-q \mu^i$, one can verify the identity $c^i_j=0$, thus, $c_i$ are functions of $R^i$ only, so that equations (\ref{P6}) take the form
\begin{equation}
\begin{split}
& \frac{\lambda^i_j}{\lambda^j-\lambda^i}=
\frac{\mu^i_j}{\mu^j-\mu^i}=
\frac{\lambda^i-\lambda^j\mu^i/\mu^j}{c^i-c^j}
\, b_j, \\
& b_{ij}=
\frac{\lambda^i(1+\mu^j/\mu^i)-\lambda^j(1+\mu^i/\mu^j)}{c^i-c^j}\, b_ib_j.
\end{split}
\label{P61}
\end{equation}
Note that by introducing potential $u$ by the formula
$u_i=\mu^ib_i$ (compatibility is guaranteed by (\ref{P61})), as well as the functions $\phi^i= \lambda^i / \mu^i$, one can reduce (\ref{P61}) to the familiar form (\ref{h7}),
$$
\phi^i_j=\frac{(\phi^j-\phi^i)^2}{c^j-c^i}\, u_j, \qquad u_{ij}=2\frac{\phi^j-\phi^i}{c^j-c^i}\, u_iu_j.  
$$

\section{Concluding remarks}

\begin{itemize}

\item We considered (integrable) PDEs $F=0$ with a (higher) conservation law $C$ with non-constant characteristic (cosymmetry) $S$, so that $\d C = FS$. Our examples suggest that the constrained system, $F=S=0$, is often involutive. It would be interesting to clarify what additional conditions are required for this to be the case.

\item Equations (\ref{V}) of the Veronese web hierarchy  describe geometric objects known as Veronese webs, see \cite{krynski2016paraconformal, ferapontov2021dispersionless}. On the contrary, equations (\ref{E}) describe potential (Egorov) metrics of diagonal curvature, typically arising in the theory of integrable systems of hydrodynamic type  \cite{tsarev1991geometry}. It seems remarkable that, put together, the equations governing these structures are compatible. In this connection, it would be interesting to understand the geometry of Veronese webs constrained by equations (\ref{E}), as well as the properties of Egorov metrics constrained by equations (\ref{V}).

\item Our examples suggest that Gibbons-Tsarev systems, governing hydrodynamic reductions of linearly degenerate dispersionless integrable PDEs, possess a Lagrangian multiform representation.

\end{itemize}

\section*{Acknowledgements}

We thank S. Agafonov, V. Caudrelier, I. Krasilshchik,  O. Morozov, V. Novikov, M. Pavlov, and Y. Suris for clarifying discussions. We also thank the reviewers for useful comments.

MV is supported by the Engineering and Physical Sciences Research Council [Project reference EP/Y006712/1].

\section*{Declarations}

\paragraph{Conflict of interest} There is no conflict of interest related to this article.

\paragraph{Data availability} A pre-print article has appeared as arXiv:2503.22615. There is no additional data related to this article.

\appendix
\section*{Appendix}
\setcounter{section}{1}
\addcontentsline{toc}{section}{Appendix}

\subsection{Lagrangian multiform (\ref{LMF}) and Euler-Lagrange equations}
\label{app-vwe}

In this section we demonstrate by direct calculation that the multiform Euler-Lagrange equations (\ref{vweEL}) for the Lagrangian multiform (\ref{LMF}) are equivalent to equations (\ref{V})--(\ref{E}). The latter can be written as $A = 0$ and $B = 0$ with
\[ A= a^i u_i u_{jk} + a^j u_j u_{ik} + a^k u_k u_{ij} , \]
where $a^i=c^j-c^k$, etc,\@ and
\[ B= 2 u_{ijk} - \frac{{u_{ij}} {u_{ik}}}{{u_{i}}} - \frac{{u_{ij}} {u_{jk}}}{{u_{j}}} - \frac{{u_{ik}} {u_{jk}}}{{u_{k}}} . \]
Note that $A = 0$ is the corner equation \eqref{vweEL3}.

\begin{lemma}
	\label{lemma-A1}
	If $A = 0$ and its differential consequences hold, then $B = 0$ is equivalent to the edge equation \eqref{vweEL2}.
\end{lemma}
\begin{proof}
	Taking into account that $a^i+a^j+a^k = 0$, the left-hand side of equation \eqref{vweEL2} can be written as
	\[ \frac{1}{u_i u_j u_k} \left( 2 A_i - a^i u_i B - \frac{u_{ij}}{u_j} A - \frac{u_{ik}}{u_k} A - 2 \frac{u_{ii}}{u_i} A \right) . \qedhere \]
\end{proof}

\begin{lemma}
	\label{lemma-A2}
	$A = 0$ and $B = 0$, together with their differential consequences, imply the planar equations \eqref{vweEL1},
	\[
	a^{k} {\left(\frac{{u_{ij}} {u_{ijj}}}{{u_{i}} u_{j}^2} - \frac{{u_{iijj}}}{{u_{i}} {u_{j}}} + \frac{{u_{iij}} {u_{ij}}}{u_{i}^2 {u_{j}}} - \frac{{u_{ii}} u_{ij}^2}{u_{i}^3 {u_{j}}} + \frac{{u_{ii}} {u_{ijj}}}{u_{i}^2 {u_{j}}} - \frac{u_{ij}^2 {u_{jj}}}{{u_{i}} u_{j}^3} + \frac{{u_{iij}} {u_{jj}}}{{u_{i}} u_{j}^2} - \frac{{u_{ii}} {u_{ij}} {u_{jj}}}{u_{i}^2 u_{j}^2}\right)} = 0,
	\]
	and equations obtained from this by cyclic permutation of $(i,j,k)$.
\end{lemma}
\begin{proof}
	Differentiate $A= 0$ with respect to $x^i$ and $x^j$ and simplify using $a^i+a^j+a^k=0$:
	\begin{align*}
		{\left({u_{i}} {u_{ijjk}} + {u_{ii}} {u_{jjk}} + {u_{iij}} {u_{jk}}\right)} a^{i}
		&+ {\left({u_{ijj}} {u_{ik}} + {u_{iijk}} {u_{j}} + {u_{iik}} {u_{jj}}\right)} a^j \\
		&+ {\left({u_{ijj}} {u_{ik}} + {u_{iij}} {u_{jk}} + {u_{iijj}} {u_{k}}\right)} a^{k} = 0 .
	\end{align*}
	We want to eliminate all derivatives with respect to $x^k$ from this equation.	
	We can use $B=0$ and its differential consequences $B_i = 0$ and $B_j = 0$ to eliminate $u_{iijk}$ and $u_{ijjk}$, and \eqref{vweEL2} to eliminate $u_{iik}$ and $u_{jjk}$. Then we simplify using $a^i+a^j+a^k = 0$ and group terms strategically to obtain
	\begin{align*}
		&\frac{a^{i}}{2} {\left(\frac{{u_{iij}}}{u_i} + \frac{{u_{ijj}} }{{u_{j}}} + \frac{u_{ij}^2 }{u_i {u_{j}}}  + \frac{2 {u_{ii}} {u_{jj}} }{u_i {u_{j}}} + \frac{{u_{ij}} {u_{ik}} }{u_i {u_{k}}} + \frac{{u_{ik}} {u_{jj}} }{{u_{j}} {u_{k}}} + \frac{{u_{ii}} {u_{jk}}}{u_i {u_{k}}} + \frac{{u_{ij}} {u_{jk}}}{{u_{j}} {u_{k}}}\right)} u_i u_{jk}  \\
		& + \frac{a^j}{2} {\left(\frac{{u_{iij}}}{{u_{i}}} + \frac{u_{ijj}}{u_j} + \frac{u_{ij}^2 }{{u_{i}} u_j} + \frac{2 {u_{ii}} {u_{jj}}}{{u_{i}} u_j} + \frac{{u_{ij}} {u_{ik}}}{{u_{i}} {u_{k}}} + \frac{{u_{ik}} {u_{jj}}}{ u_j {u_{k}}} + \frac{{u_{ii}} {u_{jk}}}{{u_{i}} {u_{k}}} + \frac{{u_{ij}} {u_{jk}}}{ u_j {u_{k}}}\right)} u_j u_{ik}  \\
		&+ \frac{a^{k}}{2} {\left( - \frac{{u_{iij}}}{{u_{i}}} - \frac{{u_{ijj}}}{{u_{j}}} + \frac{u_{ij}^2}{{u_{i}} {u_{j}}} + \frac{2 {u_{ii}} {u_{jj}} }{{u_{i}} {u_{j}}} + \frac{{u_{ij}} {u_{ik}}}{{u_{i}} u_k} + \frac{{u_{ik}} {u_{jj}}}{{u_{j}} u_k} + \frac{{u_{ii}} {u_{jk}}}{{u_{i}} u_k} + \frac{{u_{ij}} {u_{jk}}}{{u_{j}} u_k}\right)} u_k u_{ij}  \\
		&+ \frac{a^{k}}{2} {\left( 2 {u_{iijj}} - \frac{2 {u_{ii}} {u_{ijj}} }{{u_{i}}} - \frac{2 {u_{iij}} {u_{jj}} }{{u_{j}}} + \frac{2 {u_{ii}} u_{ij}^2 }{u_{i}^2} + \frac{2 u_{ij}^2 {u_{jj}} }{u_{j}^2} + \frac{2 {u_{ii}} {u_{ij}} {u_{jj}}}{{u_{i}} {u_{j}}} \right)} u_k = 0 .
	\end{align*}
	Applying $A = 0$ now allows us to eliminate $a^i$ and $a^j$:
	\begin{align*}
		& a^{k} \left( - \frac{{u_{iij}} u_{ij}}{{u_{i}}} - \frac{{u_{ijj}} u_{ij}}{{u_{j}}} \mathrel{+}  {u_{iijj}} - \frac{{u_{ii}} {u_{ijj}} }{{u_{i}}} - \frac{ {u_{iij}} {u_{jj}} }{{u_{j}}} + \frac{{u_{ii}} u_{ij}^2 }{u_{i}^2} + \frac{u_{ij}^2 {u_{jj}} }{u_{j}^2} + \frac{{u_{ii}} {u_{ij}} {u_{jj}}}{{u_{i}} {u_{j}}} \right) u_k = 0 .
	\end{align*}
	Finally, we divide by $-u_i u_j u_k$ to find the planar Euler-Lagrange equation \eqref{vweEL1}.
\end{proof}

Together, these Lemmas show that the system $A=0$, $B = 0$ is equivalent to  the system \eqref{vweEL}.

Analogous calculations show that relations (\ref{Vx}) and (\ref{E}) are equivalent to the multiform Euler-Lagrange equations in the translationally non-invariant case (\ref{LMF1}). The multiform Euler-Lagrange equations in this case are 
\begin{subequations} 
	\label{vweEL-nonautonomous}
	\begin{align}
		&(x^i-x^j) \left( \left(\frac{u_{ij}^2}{u_i^2u_j}\right)_i + \left(\frac{u_{ij}^2}{u_iu_j^2}\right)_j + \left(\frac{2u_{ij}}{u_iu_j}\right)_{ij} \right) \notag\\
		&\hspace{3cm} + \frac{u_{ij}^2}{u_i^2u_j} - \frac{u_{ij}^2}{u_iu_j^2} + \left(\frac{2u_{ij}}{u_iu_j}\right)_{j} - \left(\frac{2u_{ij}}{u_iu_j}\right)_{i} = 0, \\
		&(x^j - x^i) \left( \frac{u_{ij}^2}{u_i u_j^2} + 2 \frac{u_{ii}u_{ij}}{u_i^2 u_j} - 2 \frac{u_{iij}}{u_i u_j}\right) \notag\\
		&\hspace{3cm} - (x^k - x^i) \left( \frac{u_{ik}^2}{u_i u_k^2} + 2 \frac{u_{ii}u_{ik}}{u_i^2 u_k} - 2 \frac{u_{iik}}{u_i u_k}\right) + \frac{2u_{ij}}{u_i u_j} - \frac{2u_{ik}}{u_i u_k} = 0,  
		\\
		& (x^i-x^j)\frac{u_{ij}}{u_iu_j}+(x^j-x^k)\frac{u_{jk}}{u_ju_k}+(x^k-x^i)\frac{u_{ik}}{u_iu_k} = 0 . 
	\end{align}
\end{subequations}
The proof of Lemma \ref{lemma-A1} does not change and the proof of Lemma \ref{lemma-A2} requires one additional step of adding $\frac{B}{2 u_j u_k} - \frac{B}{2 u_i u_j}$.

\subsection{Lagrangian multiform (\ref{LMF}) and  sigma model equations}
\label{app-sigma}

Here we provide some more detail on equation (\ref{vweEL1}), setting $(i, j)=(1, 2)$:
$$
\left(\frac{u_{12}^2}{u_1^2u_2}\right)_1+\left(\frac{u_{12}^2}{u_1u_2^2}\right)_2+\left(\frac{2u_{12}}{u_1u_2}\right)_{12}=0.
$$
This fourth-order PDE can be rewritten as a second-order system,
$$
u_{12}-p\,u_1u_2=0, \qquad p_{12}+p\,p_1u_2+p\,p_2u_1+p^3u_1u_2=0,
$$
with the action functional
$$
\int \left( p^2u_1u_2+p_1u_2+p_2u_1 \right) \d x^1\wedge \d x^2.
$$
This Lagrangian governs harmonic maps from the pseudo-Euclidean ($x^1, x^2$)-plane with  the flat metric $\d x^1 \, \d x^2$, to the pseudo-Riemannian manifold with coordinates $u, p$ and the metric $p^2 \, \d u^2 + 2 \, \d u \, \d p$. Note that the latter metric has constant curvature $1$; the nonzero Christoffel symbols are $\Gamma^1_{11}=-p, \ \Gamma^2_{12}=\Gamma^2_{21}=p, \ \Gamma^2_{11}=p^3$, where we label $u, p$ as the first and second coordinates, respectively. 

In general,  sigma-models describing harmonic maps from the pseudo-Euclidean ($x^1, x^2$)-plane with  the metric $\d x^1 \, \d x^2$ to a pseudo-Riemannian manifold with coordinates $u^i$ and the metric $g_{ij} \, \d u^i \, \d u^j$, are governed by the Lagrangian $\int g_{ij}u^i_1u^j_2 \, \d x^1\wedge \d x^2$, with the Euler-Lagrange equations
$$
u^i_{12}+\Gamma^i_{jk}u^j_1u^k_2=0,
$$
where $\Gamma^i_{jk}$ are the Christoffel symbols of $g_{ij}$.

\interlinepenalty=10000

\bibliographystyle{abbrvnat_mv}
\bibliography{Fer-Ver}

\end{document}